\documentclass[conference]{IEEEtran}
\IEEEoverridecommandlockouts
\usepackage[top=2cm, bottom=2cm, left=2cm, right=2cm]{geometry}
\usepackage{url,diagbox,amsmath,amsthm,amssymb,graphicx,color,cite,algorithm,algpseudocode,algorithmicx}
\newtheorem{theorem}{Theorem}
\newtheorem{definition}{Definition}
\floatname{algorithm}{Algorithm}

\begin{document}
\title{Matrix Bloom Filter: An Efficient Probabilistic Data Structure for 2-tuple Batch Lookup}
\author{
\IEEEauthorblockN{Yue Fu\IEEEauthorrefmark{1}\IEEEauthorrefmark{3}, Rong Du\IEEEauthorrefmark{1}\IEEEauthorrefmark{3}, Haibo Hu\IEEEauthorrefmark{3}, Man Ho Au\IEEEauthorrefmark{2}, Dagang Li\IEEEauthorrefmark{1}}\\
\IEEEauthorblockA{\IEEEauthorrefmark{1}School of ECE, Shenzhen Graduate School, Peking University, China}\\
\IEEEauthorblockA{\IEEEauthorrefmark{2}Department of Computing, The Hong Kong Polytechnic University, Hong Kong, China}\\
\IEEEauthorblockA{\IEEEauthorrefmark{3}Department of Electronic and Information Engineering, The Hong Kong Polytechnic University, Hong Kong, China}\\
Email: fuyuefyu@126.com, rongdu@pku.edu.cn, haibo.hu@polyu.edu.hk\\  csallen@comp.polyu.edu.hk, dagang.li@ieee.org
}
\maketitle

\begin{abstract} 
With the growing scale of big data, probabilistic structures receive increasing popularity for efficient approximate storage and query processing. For example, Bloom filters (BF) can achieve satisfactory performance for approximate membership existence query at the expense of false positives. However, a standard Bloom filter can only handle univariate data and single membership existence query, which is insufficient for OLAP and machine learning applications. In this paper, we focus on a common multivariate data type, namely, 2-tuples, or equivalently, key-value pairs. We design the matrix Bloom filter as a high-dimensional extension of the standard Bloom filter. This new probabilistic data structure can not only insert and lookup a single 2-tuple efficiently, but also support these operations efficiently in batches --- a key requirement for OLAP and machine learning tasks. To further balance the insertion and query efficiency for different workload patterns, we propose two variants, namely, the maximum adaptive matrix BF and minimum storage matrix BF. Through both theoretical and empirical studies, we show the performance of matrix Bloom filter is superior on datasets with common statistical distributions; and even without them, it just degrades to a standard Bloom filter. 

\end{abstract}

\section{Introduction}With the advent of Internet of things, all devices around us are generating data in day-to-day life. The explosion of data and wide-adoption of machine intelligence result in billions of interactive or automated queries issued and executed per day on large web servers or databases. Due to the ever-increasing scale of data, probabilistic data structures that store raw data and execute queries in an approximate manner becomes more and more popular. A classic one is the Bloom filter\cite{1}, which was first proposed by Bloom in 1970. It is a bit array (initially set to 0) of length $m$, and is defined on $k$ uniform and independent hash functions. To insert an element, BF sets the $k$ positions from hash functions to 1. To query whether an element has been inserted, BF simply checks if all the $k$ positions corresponding to this element are set to 1. 

The standard Bloom filter can only handle univariate data. However, many data in modern OLAP and machine learning applications are multivariate. In this paper, we focus on a common multivariate data type, namely, 2-tuples, or equivalently, key-value pairs. Besides its wide applications in database and data mining, 2-tuple is also a popular data type in other fields. For example, in a content distribution network (CDN) the nodes (i.e., servers) are labelled as (A,B,C,...), and the contents are numbered by (1,2,3,...). A key problem in CDN is to determine whether a node has a copy of a content. This is equivalent to a lookup of a 2-tuple (node name, content number) in the CDN metadata. As another example, most log files in operating systems and web servers record a timestamp together with the log event. As such, each log entry can be considered as a 2-tuple in the form of (event, timestamp). 

To support 2-tuple insertion and query, we are tempted to adapt some prior work. Guo et al. proposed multi-dimensional Bloom filter (MDBF) for insertion and query of an $l$-dimensional element $(x_{1},x_{2},...,x_{l})$\cite{2}. MDBF allocates $l$ different standard Bloom filters of the same length, and insert each dimension into its corresponding Bloom filter. When querying an element, MDBF simply looks up each dimension in its corresponding Bloom filter. However, MDBF is not suitable for 2-tuple insertion and query. First and foremost, MDBF only performs membership test on each dimension independently, without keeping the correlation (e.g. co-existence) across dimensions. A key-value pair, however, needs the correlation to be stored and queried. Second, MDBF also leads to huge performance inefficiency because a false positive on any dimension always results in a false positive on the entire element. Last but not the least, MDBF cannot support both operations efficiently in batches, which is required in many OLAP and machine learning tasks. For instance, in CDN it is common to query on multiple contents such as ``Does node A cache contents 1, 4 and 6?". Likewise, a web administrator may want to know if three IPs 58.61.50.73, 152.168.50.76, and 125.153.42.12 have accessed the server from 9:00 am - 10:00am. To support batch operations, MDBF has to treat duplicate keys or duplicate values as special elements and insert them to different positions from where the first key or value is inserted. This makes the batch insertion and query inefficient.

To the best of the author's knowledge, there is no prior study on designing a probabilistic structure that could carry out batch lookups of 2-tuples. In this paper, we propose the matrix Bloom filter as an efficient solution to this problem. In a nutshell, it is a high-dimensional extension of the standard Bloom filter that performs insertions and queries on 2-tuples. Then we further propose two variants, namely, the minimum storage matrix and maximum adaptive matrix BF, that can balance the insertion and query efficiency for different workload patterns, and further exploit common statistical distributions of datasets. The contributions of this paper are summarized as follows:
\begin{enumerate}
\item \textbf{Matrix Bloom filter.} We present a unified framework to batch process 2-tuple partial existence queries by a novel data structure called matrix Bloom filter. It degrades to a standard Bloom filter for univariate data.
\item \textbf{Minimum storage/maximum adaptive matrix.} We propose two variants in the framework of matrix Bloom filter, i.e. minimum storage matrix and maximum adaptive matrix, for datasets with specific statistical features.
\item \textbf{Empirical studies.} We experimentally verify the correctness and effectiveness of matrix Bloom filter against some baseline approaches on two real-world datasets and one synthetic dataset.
\end{enumerate}

The rest of the paper is organized as follows. Section \uppercase\expandafter{\romannumeral2} formally defines the problem and introduces some baseline approaches. Section \uppercase\expandafter{\romannumeral3} proposes the framework of matrix Bloom filter. Section \uppercase\expandafter{\romannumeral4} proposes the maximum adaptive matrix and the minimum storage matrix in this framework. Section \uppercase\expandafter{\romannumeral5} discusses the experimental design and results. Section \uppercase\expandafter{\romannumeral6} reviews the related work. Section \uppercase\expandafter{\romannumeral7} draws the conclusion.

\section{Problem definition}
\subsection{2-tuple partial existence tests}As mentioned, a traditional Bloom filter does not perform batch processing efficiently, as it treats a multi-dimensional input as an entirety and loses the capability to identify common components that two different inputs may have. On the other side, even though MDBF\cite{2} is friendly for queries on a single component, it could still not solve this problem, since it loses the dependency between components, and thus is not able to tell whether those components belong to a same element. Hence, what we need is a methodology that is able to perform insertions/queries on each component orthogonally, without loss of their dependencies. We abstract this requirement as the following definition\footnote{For the convenience of discussion, only 2-D case is discussed in this paper. The extension to higher dimension is natural.}:

\begin{definition}
\textbf{2-tuple partial existence tests.} Given a 2-tuple denoted by $(x_{1},x_{2})$. A data structure is said to be able to perform 2-tuple partial existence tests, if there exists an insertion/query pattern, such that insertions/queries on $x_{1}$ and $x_{2}$ are conducted independently, while the dependency between components remains.
\end{definition}

Apparently, 2-tuple partial existence tests are friendly to batch queries on one component, since the queries are only to be operated once on the repeating component. Here are some existing typical examples that can be concluded into 2-tuple partial existence tests:\\
\textbf{Example 1. Key-value search.} Generically, hundreds of values can be corresponded to a key. Usually, we locate the key-value pairs through the key, and then search on the values, which is essentially a batch 2-tuple partial existence test.\\
\textbf{Example 2. Tmt-query.} Peng et al. proposed a question in literature \cite{3}: When an element comes in the form of (IP address, time), how to answer such question ``Is there any packet comes from a given IP address in the range of time $(t_{1},t_{2})$?"

Formally, the authors defined the notion of temporal membership testing: Given an upper bound $\mathbb{T}$ on the time dimension, and a universe $\mathbb{U}$ where elements are drawn from, a temporal set $\mathbb{A}$ consist from (element, timestamp) pairs up to time T (assuming discrete timestamps). Let $A\in\mathbb{A}$, and $A[s,e]$ refers to the subset of distinct elements in $\mathbb{A}[s,e]$. A temporal membership testing query (tmt-query) asks if $x\in A[s,e]$ for any query element $x$ chosen from $\mathbb{U}$.

It is easy to see, tmt-query fixes in the IP address, and then search on the range of a given discrete time interval, which is also essentially a batch 2-tuple partial existence test. 

\textbf{$\bullet$ Baseline approach: Using a hashmap.} A naive approach to answer partial existence tests on key-value pairs is the hashmap. Given a set of key-value pairs holding the same key and varying values, it is straightforward to store them in a hashmap structure. When queries are performed, the key is only required to be hashed a single time to find its bucket, which notably reduces the overall workload\footnote{A traditional Bloom filter is not able to work in this pattern, since two elements shall be mapped into two totally random positions, even if they hold a same key.}. See Fig. \ref{fig1} as an example of hashmap approach. A key-value pair finds its bucket according to the hash output of its key, and values with the same key are mapped into the same buckets. The items in a bucket are connected with a chain. The red item in the figure is the case of hash collision, which implies the hash output of key1 and key2 is the same.

\begin{figure}[!t]
  \centering
  \includegraphics[height=1.9in]{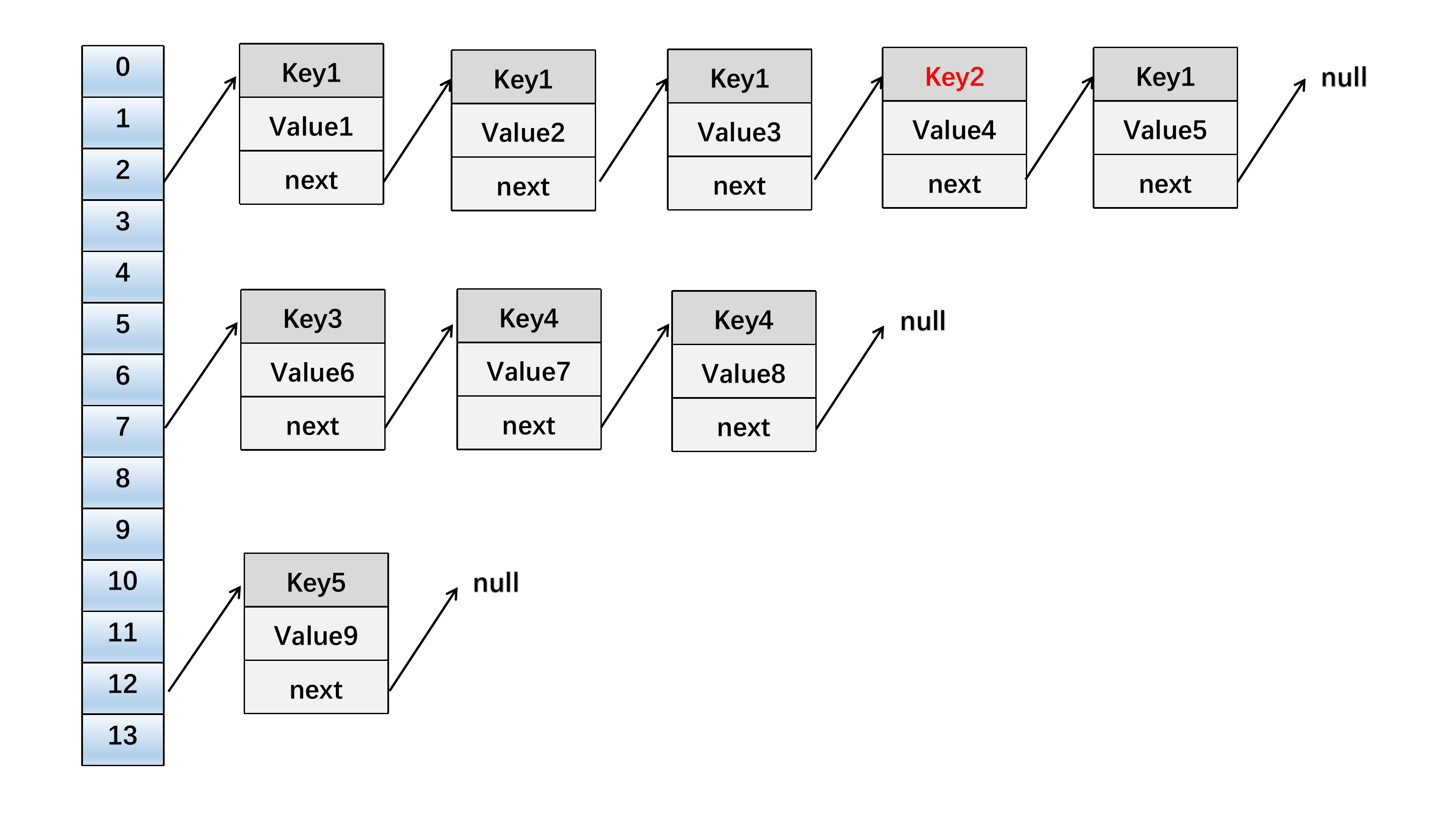}
  \caption{An example of hashmap approach.}
  \label{fig1}
\end{figure}

The disadvantage of the hashmap approach is three-folded:

\begin{enumerate}
\item[$\bullet$]A hashmap store the key-value pair in its full-length, which can be space-consuming.
\item[$\bullet$]Since the values are to be stored by chaining, the search can be expensive with the scale increasing (even if we employ a b-tree). 
\item[$\bullet$]It does not support double-side partial queries.
\end{enumerate}

\subsection{Double-side partial query}So far, we have discussed the case of key-value search, where keys are fixed and values are to be batch processed. However, in practical scenarios, the borderline of ``key" and ``value" can be vague. That is, given a set of key-value pairs with a fixed ``value", a batch query can also be processed on ``keys". For instance, in the context of CDN, one may either ask ``Does node A owns content 1, 4 and 6?", or ``Are nodes A, C and F the owners of content 3?" 
Apparently, this leads to lookups on element (A,1), (A,3), (A,5), and (A,3), (C,3), (F,3), respectively. In the end of literature\cite{3}, the authors also proposed an open question: ``How to extend the one-way approximate query into the two-way approximate query?" i.e. ``given an IP set $(IP_{1},IP_{2},...,IP_{n})$ as well as a time range $(t_{1},t_{2})$, is there any packet comes from somewhere close to an IP address in the given IP set in the given time range?" 

To formally describe the problem, we introduce the notion of double-side partial query to formulate the demand of batch queries on each possible component:
\begin{definition}
\textbf{Double-side partial query.} Given a 2-tuple denoted by $(x_{1},x_{2})$. A data structure is said to support double-side partial queries, if there exists an insertion/query pattern, such that the dependency between $x_{1}$ and $x_{2}$ holds, and the overall result of the entire element $(x_{1},x_{2})$ is unconcerned with the query sequence of $x_{1}$ and $x_{2}$.
\end{definition}

\textbf{$\bullet$ Baseline approach: Using two hashmaps.} To support double-side partial queries, a straightforward solution is to commute all the key-value pairs, then store the generated value-key pairs into another hashmap. When performing batch queries on a fixed value, employs the latter hashmap, vice versa. See Fig. \ref{fig2} for an example of the two hashmap solution. Needless to say, this solution is quite cumbersome, as the storage cost would be even doubled than using a single hashmap. 

\begin{figure}[!t]
  \centering
  \includegraphics[height=2in]{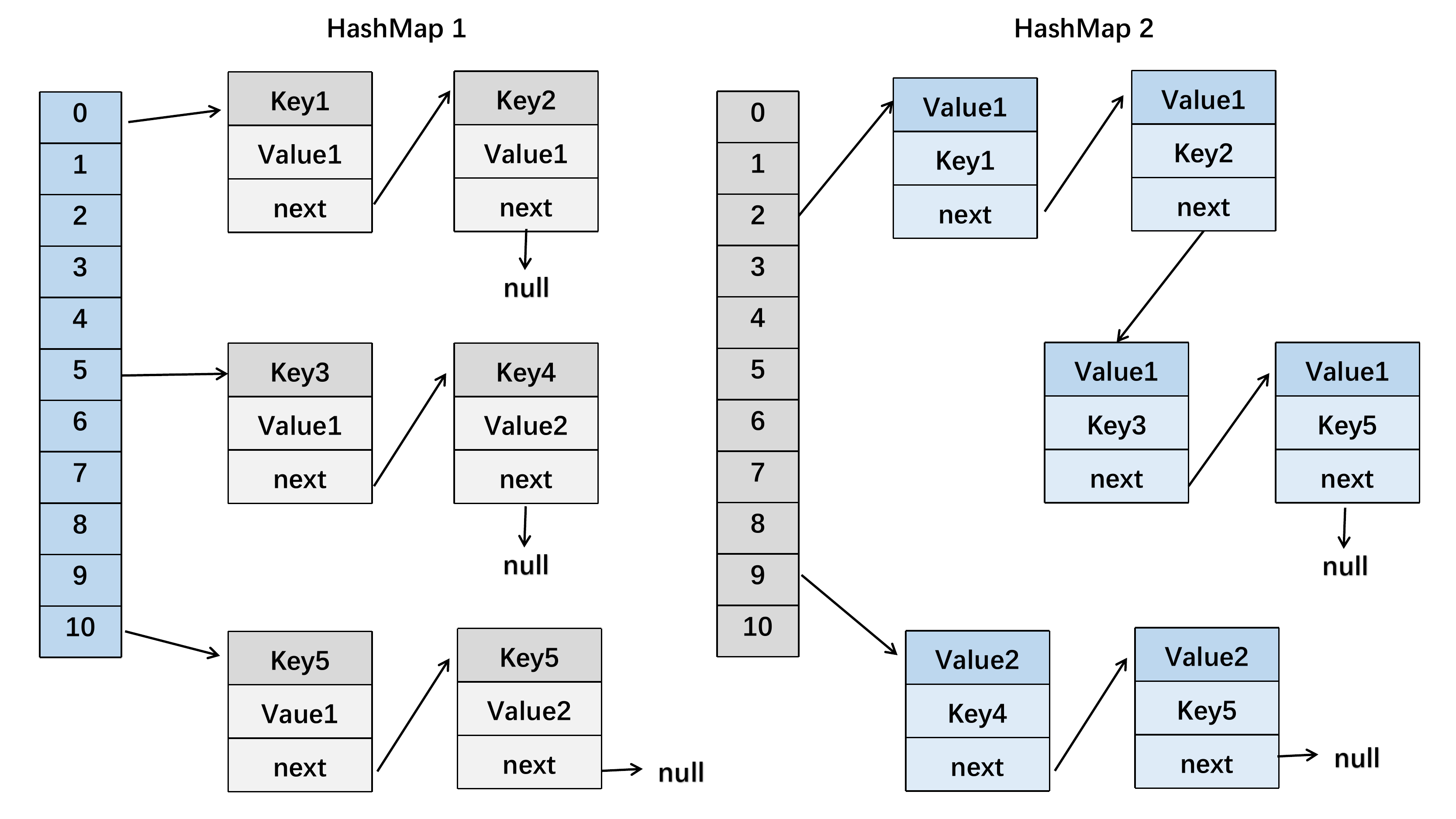}
  \caption{An example of double-hashmap approach.}
  \label{fig2}
\end{figure}

\section{Matrix Bloom Filter: A novel framework}In the last section, we have discussed some hashmap-based solutions that store raw data in their full-length. In this section, we propose a Bloom-filter based probabilistic structure to process 2-tuples with constant space and small errors. 
\subsection{Intuition}Let's recall the structure of a standard Bloom filter, which employs a bit-vector to perform insertions/queries on scalar elements. Now, the form of inserting elements turns out to be vector $(x_{1},x_{2})$, and we aim at developing a bit-matrix to perform batch 2-tuple partial existence tests on the vectors. It would be a nature extension of the standard Bloom filter, thus the operations on components shall also be similar.

Intuitively, the very beginning is a bit-matrix, whose bits are initially set to 0. Each row is treated as a standard 0/1 array Bloom filter with parameter $(m_{1},n_{1},k_{1})$, as well as each column with parameters $(m_{2},n_{2},k_{2})$. Suppose an element $A_{ij}(x_{i}^{1},x_{j}^{2}) \in \mathbb{S}$ is to be inserted into our matrix, where $n=|\mathbb{S}|$. Let $x_{i}^{1} \in \mathbb{S}_{1}$ and $x_{j}^{2} \in \mathbb{S}_{2}$, hence $n_{1}=|\mathbb{S}_{1}|$ and $n_{2}=|\mathbb{S}_{2}|$. Clearly, $max(n_{1},n_{2}) \leq n \leq n_{1}n_{2}$. Overall, $m_{0}=m_{1}m_{2}$ bits and $k_{1}+k_{2}$ hash functions are used to store $n$ elements from set $\mathbb{S}$.
\subsection{Element insertion/query}Consider an element chosen from $\mathbb{S}$, say, $A(x_{1},x_{2})$. As shown in Fig. \ref{fig3}, 
firstly, hash the $(x_{1})$ to $k_{1}$ row indices through the $k_{1}$ row hash functions. Next, hash the $(x_{2})$ to $k_{2}$ column indices through the $k_{2}$ column hash functions. Finally, place the $k_{0}=k_{1}k_{2}$ bits located at intersections of selected $k_{1}$ rows and $k_{2}$ columns to 1. The algorithm is shown in Algorithm \ref{alg1}.

\begin{figure}[!t]
  \centering
  \includegraphics[height=2.1in]{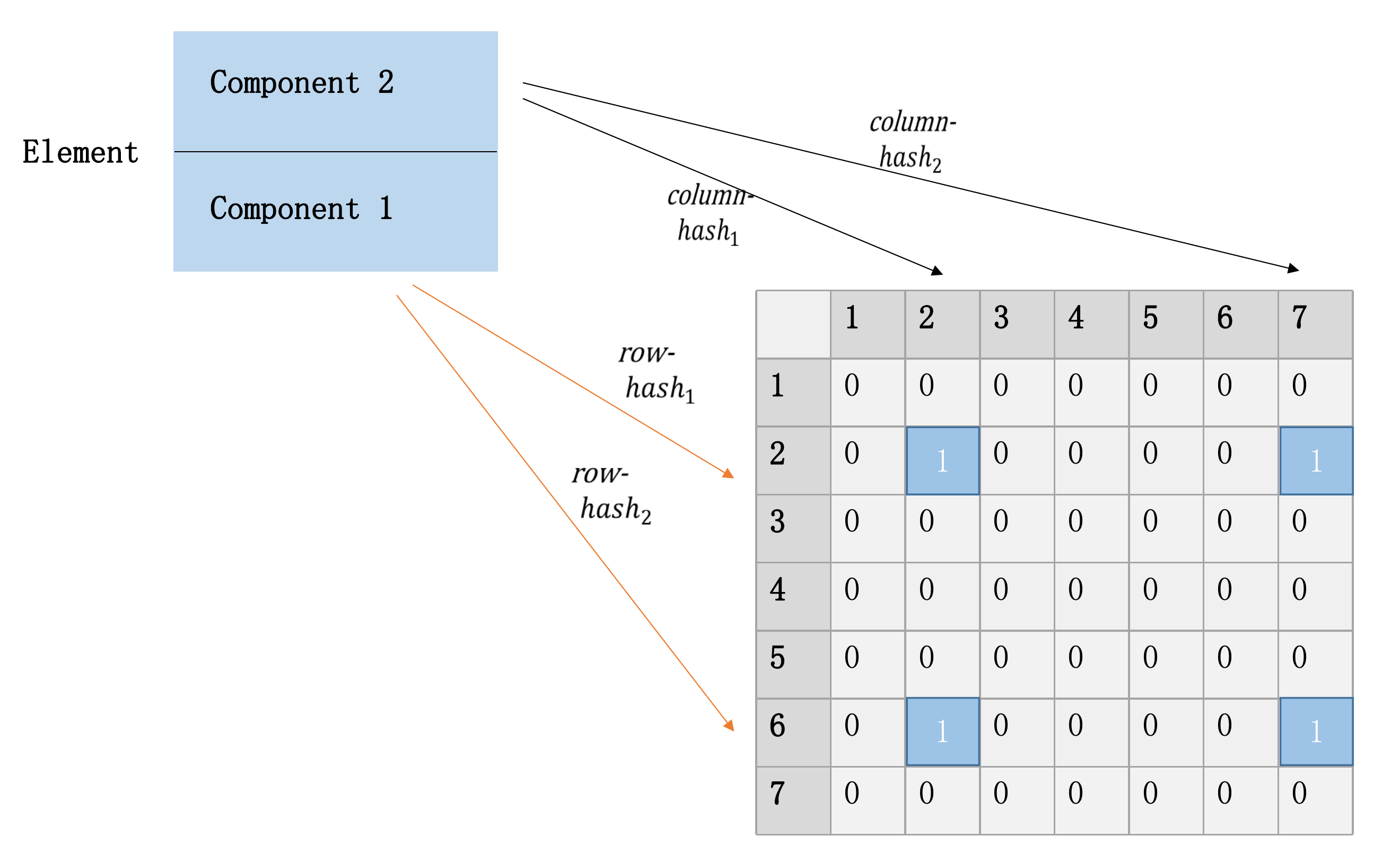}
  \caption{Insertion of an 2-dimensional element into matrix Bloom filter. Components are hashed into row/column indices and their intersections are set to 1.}
  \label{fig3}
\end{figure}

\begin{algorithm}[!h]
\caption{Element insertion of a matrix Bloom filter}
\label{alg1}
\begin{algorithmic}[1]
\State Input $A(x_{1},x_{2})$;
\State For i=1 to $k_{1}$; 
\State For j=1 to $k_{2}$;
\State $row\_array[i] \leftarrow rowhash_{i}(x_{1})$;
\State $column\_array[j]  \leftarrow columnhash_{j}(x_{2}) $;
\State $matrix(row\_array[i],column\_array[j]) \leftarrow 1$;
\State Endfor;
\State Endfor;
\end{algorithmic}
\end{algorithm}

\indent When any $A(x_{1},x_{2})$ is queried, hash $x_{1}$ and $x_{2}$ to the corresponding $k_{1}k_{2}$ positions, and check if all of them are 1. If yes, return ``positive", otherwise return ``negative", as described in Algorithm \ref{alg2}.

\begin{algorithm}[!h]
\caption{Element query of a matrix Bloom filter}
\label{alg2}
\begin{algorithmic}[1]
\State Input $A(x_{1},x_{2})$;
\State For i=1 to $k_{1}$; 
\State For j=1 to $k_{2}$;
\State $row\_array[i] \leftarrow rowhash_{i}(x_{1})$;
\State $column\_array[j]  \leftarrow columnhash_{j}(x_{2}) $;
\State If any $matrix(row\_array[i],column\_array[j])==0$ then outputs ``negative";
\State break;
\State EndIf
\State Endfor;
\State Endfor;
\State outputs ``positive";
\end{algorithmic}
\end{algorithm}

\subsection{False Positive Rate}Clearly, when some element is queried in the matrix Bloom filter, it returns false positive results, when all the $k_{1}k_{2}$ bits are set to 1, but the corresponding element is in fact not included in the membership set. In this part, let's derive a series of significant theoretical results on our matrix Bloom filter, based on the logic of standard Bloom filter. 

Firstly, let's determine the theoretical false positive rate on matrix Bloom filters:
\begin{theorem}False positives occur in matrix Bloom filters, with probability $(1-e^{\frac{-nk_{1}k_{2}}{m_{1}m_{2}}})^{k_{1}k_{2}}$.
\label{theorem1}
\end{theorem}

\begin{proof}
See Appendix A.
\end{proof}

Next, let's decide the optimal number of hash functions and the lowest false positive rate of matrix Bloom filters: 

\begin{theorem}Given the relationship $k=\frac{m}{n}ln2$ as the theoretical optimal condition of standard Bloom filters. Simply substitute $k=k_{1}k_{2}, m=m_{1}m_{2}$ into this formula, we have the optimal condition of our matrix Bloom filter.
\end{theorem}

\begin{proof}
See Appendix B.
\end{proof}

\subsection{Properties}Upon the discussions above, we are able to claim several significant properties of a matrix Bloom filter. 

Let's start with the discussion on false negatives, that is, the output of an element is reported to be negative in a membership test, but the element in fact belongs to the member set. It is easy to see, the insertion/deletion rules in each position are same to the standard Bloom filter, thus we have the assertion:

\textbf{Assertion 1.} \textit{One side error occurs in the matrix Bloom filter, that is, no false negatives would appear, just as a standard Bloom filter does.}

The operation on a matrix Bloom filter is componentwise, while their dependency remains. Furthermore, the final result is irrelevant to the query sequence of the components. Thus, we have the following two assertions:

\textbf{Assertion 2.} \textit{A matrix Bloom filter supports batch 2-tuple partial existence tests.}

\textbf{Assertion 3.} \textit{A matrix Bloom filter supports double-side partial queries.}

Compared with a standard Bloom filter of size $m$, if the same amount of bits $m=m_{1}m_{2}$ is used to build a matrix Bloom filter, the same theoretical optimal values/boundary conditions can be achieved. It is already shown matrix Bloom filters process batch query on one component more effectively. If both of the two components vary concurrently, the overall performance degrades. A natural extension of theorem 2 can be made to claim the following assertion:

\textbf{Assertion 4.} \textit{If queries on elements are no longer batchwise, e.g. both components of the inserting 2-tuples are randomly given, the performance of the matrix Bloom filter degrades to a standard Bloom filter of the same size.}

Furthermore, it is easy to extend the notion of matrix Bloom filter to other classes in the family of Bloom filters. For example, the notion of counting Bloom filter\cite{4} can be introduced directly to make elements deletable. Matrix Bloom filter can also be compressed during the transmission processes to achieve space-efficiency\cite{5}. However, detailed discussions are beyond the scope of this paper.

\section{Two variants of matrix Bloom filter}
\subsection{Background knowledge on datasets}In the last section, we have shown a matrix Bloom filter supports both 2-tuple partial existence tests and double-side partial queries. When a set of 2-tuples is queried, if a value on one component occurs multi-times, the corresponding hashing results can be immediately reused in the following process of other elements, without additional calculations. For instance, suppose there are five 2-tuples ${(1,2),(1,3),(1,4),(2,5),(3,6)}$, and we are performing query on them according to the rule of matrix Bloom filter. The first component 1 should be hashed to find out the row-indices we need, thus, for the first three elements, it is unnecessary to redo it two extra times. Hence, if we have a strong background knowledge on the datasets, our matrix Bloom filter performs batch queries more efficiently.

However, without knowledge on the statistical feature of inserting datasets, the problem can be complex. To illustrate this point, let's look equation.\ref{2} in appendix A, which is proved based on correctness of equation.\ref{1}. When the dataset holds some specific statistic distribution rather than uniformly given, correctness of equation.\ref{1} may no longer remains. As an example, let's make our discussion on a specific case: Consider a dataset constitutes from 10 elements, where all the $x_{1}$s are fixed to be 1, while $x_{2}$s vary. If a square matrix designed for 10 elements is allocated to store them, the rows where $x_{1}$s are hashed in will be much concentrated due to insertion of different $x_{2}$s. Clearly, equation.\ref{1} is no longer available in this case. Looking back to boundary conditions of theorem 1, the first thing remained to be discussed turns out to be: When the total amount of bits $m_{0}=m_{1}m_{2}$ and hash function $k_{1},k_{2}$ are given, how to allocate the exact proportion of these parameters?

In this section, we propose two different types of matrices working for different scenarios: the maximum adaptive matrix and the minimum storage matrix. The former one is adaptive to the the most general case with a higher storage cost. The latter one works under some specially given datasets and achieves the lowest possible storage cost (same to a standard Bloom filter).

\subsection{Maximum adaptive matrix}Notice that the structure of our matrix is not simply determined by the number of elements in set $\mathbb{S}$ or even $\mathbb{S}_{1}$ and $\mathbb{S}_{2}$ (i.e. $n,n_{1},n_{2}$), but the way how they are combined. Hence, for complex datasets, we need more information to decide a reasonable shape and size of the matrix structure. 

A natural solution is to consider the worst case, that is, for any given $n_1, n_2$, simply pre-allocate large enough spaces for any given $\mathbb{S}$ with parameters $n_1, n_2$, where the overall number of their combinations shall be at most $n_{1}n_{2}$. To be adaptive to the most general case, given $m_{1}$ and $m_{2}$, treat the rows and columns as dedicated standard Bloom filters, both of which are expected to hold the lowest possible false positive rate. Hence, we have 

\begin{equation*}
k_{1}=\frac{m_{1}}{n_{1}}ln2, k_{2}=\frac{m_{2}}{n_{2}}ln2
\end{equation*}
  
\textbf{Load factor.} The maximum adaptive matrix is relatively empty, see Fig. \ref{fig4} as a sample. Approximately 50\% rows/columns are set to 0, hence the load factor is approximately 25\%. In the worst case where $n=n_{1}n_{2}$, since the query in row and columns fulfils equation. \ref{2}, the load factor turns out to be approximately $50\% \times 50\%=25\%$.

In general cases, $n$ is in fact less than $n_{1}n_{2}$, thus the load factor is always less than $25\%$.

\begin{figure}[!h]
  \centering
  \includegraphics[height=2in]{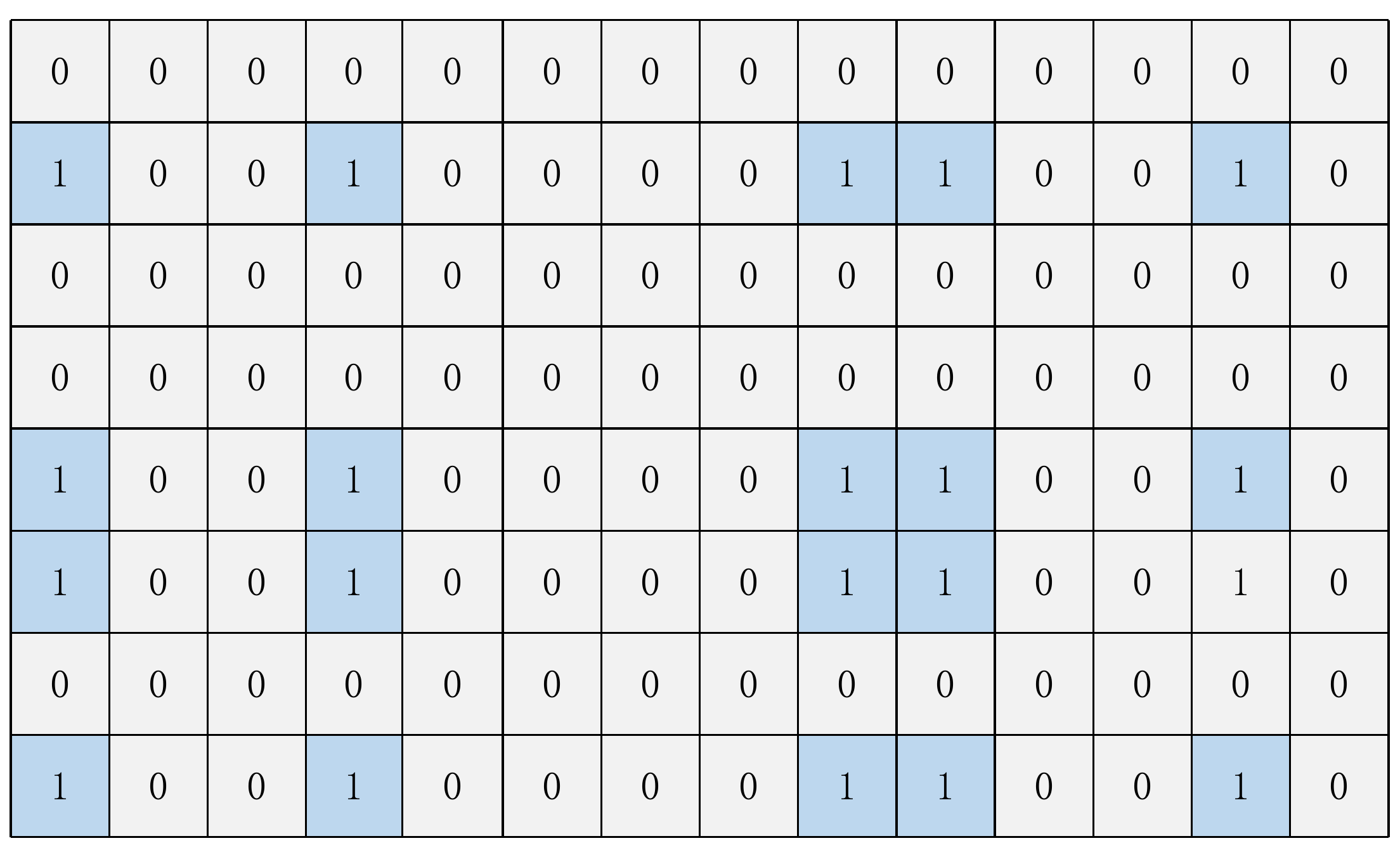}
  \caption{A sample of maximum adaptive matrix. The white rows/columns are entirely set to 0.}
  \label{fig4}
  \end{figure}

\textbf{False positive rate.} When queries in both rows and columns turn out to be false positive, the overall false positive result occurs. Hence,

\begin{equation*}
fpr_{mam}=f_1 \times f_2=(\frac{1}{2})^{k_{1}+k_{2}}
\end{equation*}

\textbf{Storage overhead.} As mentioned, the maximum adaptive matrix is relatively empty. It achieves a better performance of adaptive batch queries at the expense of a storage overhead. Let's derive the largest possible storage cost of the maximum adaptive matrix, and compare it with a standard Bloom filter.

Assume a standard Bloom filter is allocated to insert the same set of $n_1 n_2$ elements, where the same false positive rate $fpr_{mam}=(\frac{1}{2})^{k_{1}+k_{2}}$ is achieved. Let the size of standard Bloom filter be $m_{0}$, hence, the corresponding number of hash functions in the standard Bloom filter turns out to be $k_{0}=\frac{m_{0}}{n_{1}n_{2}}ln2$. Let $k_0=k_1+k_2$, we have:

\begin{equation*}
(\frac{m_{1}}{n_{1}}+\frac{m_{2}}{n_{2}})ln2=(\frac{n_{2}m_{1}+n_{1}m_{2}}{n_{1}n_{2}})ln2=(\frac{m_{0}}{n_{1}n_{2}})ln2
\end{equation*}

That is,

\begin{equation*}
n_1 m_2+n_2 m_1=m_0
\end{equation*}

Due to the mean value inequality, we have:

\begin{equation*}
2\sqrt{n_{1}n_{2}} \cdot \sqrt{m_{1}m_{2}} \leq m_{0}
\end{equation*}

Hence,

\begin{equation}\label{4}
m=m_{1}m_{2} \leq \frac{(m_{0})^{2}}{4n_{1}n_{2}}
\end{equation}

\textbf{Complexity.} For each lookup, there are overall $k_{1}k_{2}$ hashing/comparisons, therefore the complexity is simply $O(k_{1}k_{2})$. Since 

\begin{equation*}
k_{1}k_{2}=(\frac{m_{1}m_{2}}{n_{1}n_{2}})ln^{2}2
\end{equation*}

Substitute by equation.\ref{4}, we have:

\begin{equation*}
k_{1}k_{2}=(\frac{m_{1}m_{2}}{n_{1}n_{2}})ln^{2}2 \leq \frac{(m_{0})^{2}}{4(n_{1}n_{2})^{2}}ln^{2}2= \frac{k_{0}^{2}}{4}
\end{equation*}

However, in the case of batch query, the complexity can be further reduced. Consider the best case of key-value pairs, where a series of varying values corresponds to a same key. For this special dataset, the key is to be hashed only once, while the values are to be hashed at most $max(k_1, k_2)$ times. Therefore, the best complexity can be lowered to $O(max(k_1, k_2))$.

\subsection{Minimum storage matrix}As discussed in the context, the maximum adaptive matrix is adaptive to the most general case, that is, even if we have neither any background knowledge on the dataset, nor the lookup pattern, the maximum adaptive matrix could still work at the expense of a relatively empty structure, which results in a large storage overhead. However, in practical use, we may be aware of, or could foresee some statistical knowledge on the lookup pattern (e.g. batch query), or the dataset itself. With these information, the situation could be much simpler, so we can allocate a matrix in the lowest possible storage. In this part, we will discuss some typical datasets with strong background knowledge, and try to construct the corresponding matrices according to our theory.

\textbf{Case A.} Let's start at a special case where there exists a bijection between $\mathbb{S}_{1}$ and $\mathbb{S}_{2}$, that is, for any two different elements, there is no repeat on both of the two components. Clearly, in this case $n_{1}=n_{2}=n$, which means the two sets of different components contain the same amount of elements.

In this special case, the answer is simple. In this case, positions of $(k_{1},m_{1})$ and $(k_{2},m_{2})$ are reciprocal, since both $(k_{1},k_{2})$ and $(m_{1},m_{2})$ are commutable. Naturally, we employ a square matrix for insertion of the $n$ elements, where $k_{1}=k_{2}=k, m_{1}=m_{2}=m$.

We further point out that the theoretical result we described in theorem 1 is suitable to the model of Case A, since there is no repeat on any of the two components. If both of the 2 components are randomly given, for any two different elements, the probability of repeating on one dimension tends to be zero when $n$ becomes large.

\textbf{Case B.} In this case, let's consider a more general situation where elements in $\mathbb{S}$ can be represented as a weak combinations from $\mathbb{S}_{1}$ and $\mathbb{S}_{2}$. For the convenience of discussion, let $n_{1}>n_{2}$. In this case, each element in $\mathbb{S}_{2}$ combines with elements in $\mathbb{S}_{1}$ $j$ times. Elements in $\mathbb{S}_{1}$ are not allowed to repeat. Hence, $n_{1}=jn_{2}$.

\textbf{Element insertion/query.} Let's start at a special case where $n_{1}=2n_{2}$. We adopt a special hash to classify components in $\mathbb{S}_{1}$ into 2 types, each belongs to either $\mathbb{S}_{11}$ or $\mathbb{S}_{12}$. In each set, there are $n_{2}=\frac{n_{1}}{2}$ elements. For each set, we employ a square matrix as described in the context of Case A.

For any element yet to insert, we at first employ the special hash to find out which square matrix it belongs to. Then, in the determined square matrix, find out the corresponding rows for component in $\mathbb{S}_{1}$ as well as columns for component in $\mathbb{S}_{2}$.

For the two square matrices, components in $\mathbb{S}_{2}$ are always mapped into the same columns. Hence, it equals to simply stick the two matrices together from left to right to get a $2m_{2}\times m_{2}$ matrix. For the general case, when $n_{1}=jn_{2}$, similarly, the special hash should choose which of the $j$ sets an element belongs to. Then, for each set a square matrix is adopted, and the $j$ square matrices are stuck together to build a $jm \times m$ $j$-matrix. See Fig. \ref{fig5} as an example.

When any element is queried, similar rules are executed to find out if the $k_{1}k_{2}$ bits are 1.

\begin{figure}[!h]
  \centering
  \includegraphics[height=2.7in]{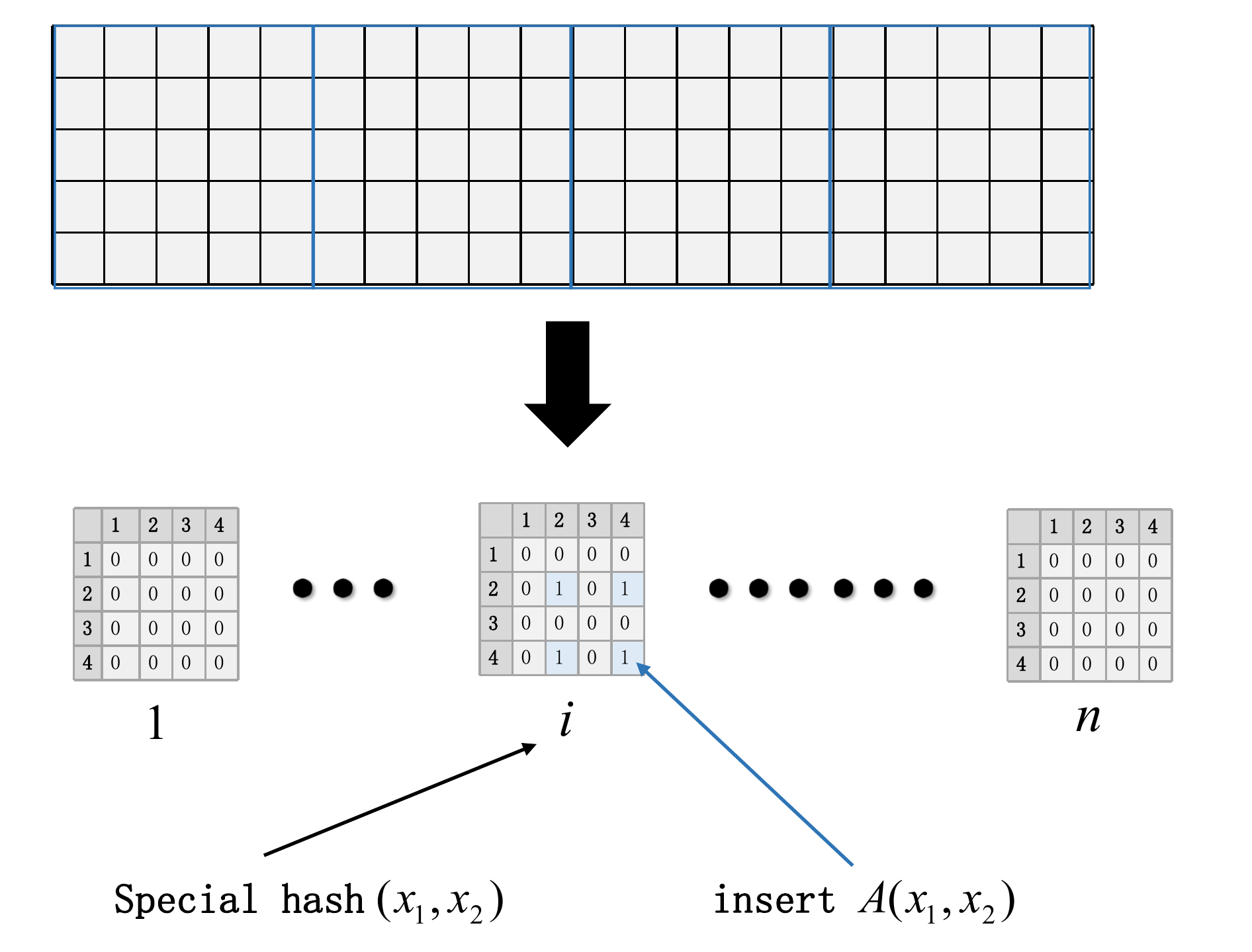}
  \caption{Element insertion/query in a j-matrix.}
  \label{fig5}
  \end{figure}
  
\textbf{False positive rate.} Firstly, let's discuss the special case where $n_{1}=2n_{2}$. Suppose overall $2m^{2}$ bits are employed to build this matrix. Now let's prove the false positive rate in this scenario equals to using a standard Bloom filter, where the same amount of bits ($2m^{2}$) is allocated to insert the $n_{1}^{2}$ elements.

Obviously, the false positive rate of that standard Bloom filter is $(\frac{1}{2})^{\frac{2m^{2}}{n_{1}^{2}}ln2}$. In the matrix, since the two square matrices shares the same false positive rate, only one of them need to be calculated. The false positive rate turns out to be:
\begin{equation*}
fpr_{2-matrix}=(\frac{1}{2})^{\frac{m^{2}}{n_{1}n_{2}}ln2}=(\frac{1}{2})^{\frac{2m^{2}}{n_{1}^{2}}ln2}
\end{equation*}

\indent Extending to the $j$-matrix case, we have the following theorem:
\begin{theorem}
When $jm^{2}$ bits are adopted to build a general $jm \times m$ $j$-matrix case, the false positive rate equals to a standard Bloom filter where the same amount of bits $(jm^{2})$ is used for insertion of the same amount of elements. (Say, $n_{1}=jn_{2}$)
\end{theorem}
\begin{proof}
The false positive rate of that standard Bloom filter is $(\frac{1}{2})^{\frac{jm^{2}}{n_{1}^{2}}ln2}$. Since the $j$ square matrices shares the same false positive rate, just one of them need to be considered. The false positive rate turns out to be:
\begin{equation*}
fpr_{j-matrix}=(\frac{1}{2})^{\frac{m^{2}}{n_{1}n_{2}}ln2}=(\frac{1}{2})^{\frac{jm^{2}}{n_{1}^{2}}ln2}
\end{equation*}
\end{proof}

\textbf{Storage cost.} Similarly, suppose a standard Bloom filter is allocated to insert the set of entire elements and the same false positive rate is achieved. Let the total amount of standard Bloom filter be $m_{0}$ and $k_{0}=\frac{m_{0}}{n_{1}n_{2}}ln2$. Let $k=k_{1}k_{2}$. Comparing with eq.\ref{2} we have $m=m_{0}, k=k_{0}$, which means it behaves no different to a standard Bloom filter.

\textbf{Complexity.} Since $k=k_{0}$, $O(k)=O(k_{0})$. When being batch queried, needless to say, the complexity is lowered to $O(max(k_1, k_2))$.

\subsection{Conclusion of this section}In this section, we introduce two variants of the matrix Bloom filter. The maximum adaptive Bloom filter is more adaptive to general cases, at the expense of a larger storage cost. The minimum storage matrix exploits some prior information to optimize the performance. The best possible storage cost is no different from using a standard Bloom filer. During batch processing, both of their complexity is lowered to $O(max(k_1, k_2))$. Detailed properties derived in this section are summarized in the following table.

\begin{table}[!h]
\centering
\caption{Comparison between two variants of matrix Bloom filter}
\begin{tabular}{|c|c|c|c|c|c|}
\hline
\diagbox{Attribute}{Type}&Maximum adaptive&Minimum storage\\
\hline
False positive rate&$(\frac{1}{2})^{k_{1}+k_{2}}$&=standard BF\\
\hline
Storage cost&$\frac{(m_{0})^{2}}{4n_{1}n_{2}}$&=standard BF\\
\hline
Load factor&$\leq25\%$&=standard BF\\
\hline
Complexity&$O(k_1k_2)$&=standard BF\\
\hline
Batch complexity&$O(max(k_1, k_2))$&$O(max(k_1, k_2))$\\
\hline
With prior knowledge&no&yes\\
\hline
\end{tabular}
\end{table}

\section{Experiments}In this section, we present the experiments of evaluating our proposed matrix Bloom filter and its two variants. Our experiments aim to verify the following things:

\begin{enumerate}
\item[$\bullet$]\textbf{False positive rates.} We derived a series of theoretical false positive rates on matrix Bloom filter (theorem \ref{theorem1}) as well as its two variants. We perform some empirical studies with real-world datasets to verify their correctness.
\item[$\bullet$]\textbf{Batch performance.} We proposed the notion of 2-tuple partial existence tests and claimed it is friendly for batch queries where one component is fixed. We also discussed using hashmap as a baseline solution. We perform the comparison on some real-world datasets to verify the effectiveness of our maximum adaptive matrix.
\item[$\bullet$]\textbf{Double-side partial query performance.} We have discussed that our matrix Bloom filter support double-side partial query and proposed using two hashmaps as a baseline solution. To verify the effectiveness of the matrix Bloom filter, we create a synthetic dataset of 2-tuples, and perform double-side partial query on the baseline solution and our maximum adaptive matrix as comparison.
\end{enumerate}

\subsection{Configurations}
\textbf{Datasets.} We employ two real-world datasets from the Bag of Words\footnote{http://archive.ics.uci.edu/ml/machine-learning-databases/bag-of-words/}, which constitutes from five text collections. We choose docword.kos (with 353160 key-value pairs from 3430 keys and 5851 values) and docword.nips (with 746316 key-value pairs from 1500 keys and 12375 values). 

We also create two types of synthetic datasets of 2-tuples. The first type is the full-repeating datasets, which includes four datasets constituted from the Euclidean cross product of two scalar datasets where elements are randomly numbers, to endow the tuple with repeating feature on both of the two components. The first scalar dataset is fixed, with number 1000. The second scalar varies, at 50, 200, 500, 1000, respectively. Hence, there are overall 50000, 200000, 500000, 1000000 2-tuples in the four datasets, respectively.

The second type is the no-repeating dataset, which contains only one dataset, which is also constituted from two scalar datasets of random numbers. However, a 1-to-1 mapping exists between the two scalar sets, that is, no repeating on both of the components for any two different 2-tuples.

As for the lookup sets, we employ two sections of data. The first section is data that are already known to exist in the member set. The second section is data that are already known NOT to exist in the member set. The overall dataset inserted/queried is a combination of the two sections.

\textbf{Hash functions.} For the Bloom filters, since a number of different hash functions is needed, we choose universal functions\cite{36} to map elements into Bloom filters. For any item $X$ with b-bits representation as  $X=<x_1,x_2,...,x_b>$. The $ith$ hash function over $X$  $h_{i}(x)$ is calculated as  $h_{i}(x)=(d_{i1}\times x_{1})\oplus(d_i2\times x_2)\oplus...\oplus(d_i3\times x_3)$, where `$\times$' is bitwise $AND$ operator and `$\oplus$' is bitwise $XOR$ operator. $d_{i}s$ are predetermined random numbers that has the same value range as the hash function.

For the hashmaps, only one hash function is needed. We simply employ the stochastic multiplier method.

\textbf{Performance metrics.} The false positive rate part is naive: We simply compare them with their theoretic values. As for the batch performance and the double-side partial query performance part, we just count how many ``compare" operations they perform. We also measure the runtime of double-side partial query as comparison to evaluate whether the matrix Bloom filter is cache-friendly.

\subsection{False positive rates}
\textbf{On generic matrix Bloom filters.} In this part, we will verify the correctness of theorem \ref{theorem1}. Firstly, let's introduce how the parameters are chosen. Since the look up pattern is not concerned in this part, we choose $n=2^{10}$ 2-tuples from the no-repeating dataset.

For a standard Bloom filter, insert the 2-tuples as an entirety. Let $k$ be different integers, and adjust $m$ to appropriate values with the formula $k=\frac{m}{n}ln2$. For a matrix Bloom filter, allocate the same amount of bits and elements where $k=k_{1}k_{2}$, as a comparison to the standard Bloom filter. For the matrix Bloom filter, the insertion rules of a single row/column is exactly the same to a standard Bloom filter.

The left part of Fig.\ref{fig6} shows the tendency of false positive rates of a standard Bloom filter varies with respect to $k$, and the right part is that of a matrix Bloom filter. Specially, the points in the matrix Bloom filter is more concentrated, as $m_1$ and $m_2$ are commutable, and renders the matrix Bloom filter be a square matrix. Notice that $k=k_{1}k_{2}$, if both $k_{1}$ and $k_{2}$ are required to be strict integers, fewer choice of $k$ is available. Hence, as an approximation, we take some points nearby $k=8$ and $k=16$ where the number of hash functions are fixed to integers, however, the corresponding $m_1$ and $m_2$ are calculated from the non-integer values of $k$. We can see from Fig.\ref{fig6} that the experimental values fit in the theoretic results. Hence, theorem \ref{theorem1} is verified, which implies the performance of our matrix Bloom filter behaves equally to the standard Bloom filter for no-repeat datasets.
  
\begin{figure}[!h]
\centering
\includegraphics[width=1.3in,angle=270]{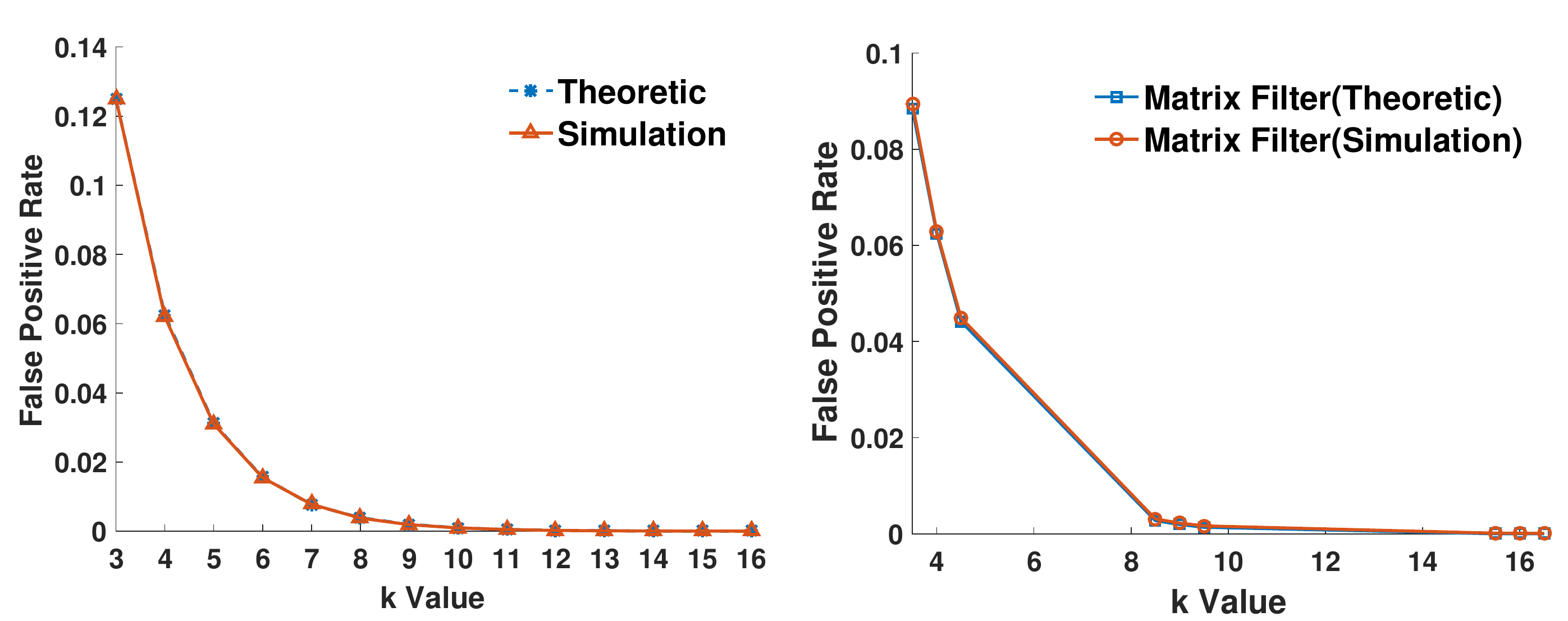}
\caption{Comparison of theoretical/experimental value of minimum false positive rate in a standard/matrix Bloom filter.}
\label{fig6}
\end{figure}

\textbf{On maximum adaptive matrix.} In this part, we will evaluate the performance of the maximum adaptive matrix. Let $k=k_{1}+k_{2}$, fix $n_{1}=256$, $n_{2}=512$ and let $k$ be different integers. $m_{i}$s are adjusted to appropriate values with the formula $k_{i}=\frac{m_{i}}{n}ln2, i=1, 2$, respectively.

Firstly let's test the false positive rate when the matrix is fully loaded. Hence, we employ the full-repeating dataset where $n=n_{1}n_{2}$. Let $\frac{m_{1}}{m_{2}}=1,\frac{1}{2},\frac{1}{4}$ respectively, and draw the tendencies of false positive rates varying with respect to $k$ on the left part of Fig. \ref{fig7}. In this picture, we can see on all the 3 curves, our maximum adaptive matrix well-fits the theoretic value.

Secondly, let's test the false positive rate when the matrix is relatively empty. Fix $k_{1}+k_{2}=6$, and adjust $m_{i}$s to appropriate values with the formula $k_{i}=\frac{m_{i}}{n}ln2, i=1,2$. The proportion of inserted elements (randomly picked from the full-repeating dataset where $n=n_{1}n_{2}$) is represented in the horizontal axis. The results are shown in the right part of Fig. \ref{fig7}. The tendency how the practical false positive rate decrease can be seen, when the matrix is not fully loaded.

\begin{figure}[!h]
\centering
\includegraphics[width=1.3in,angle=270]{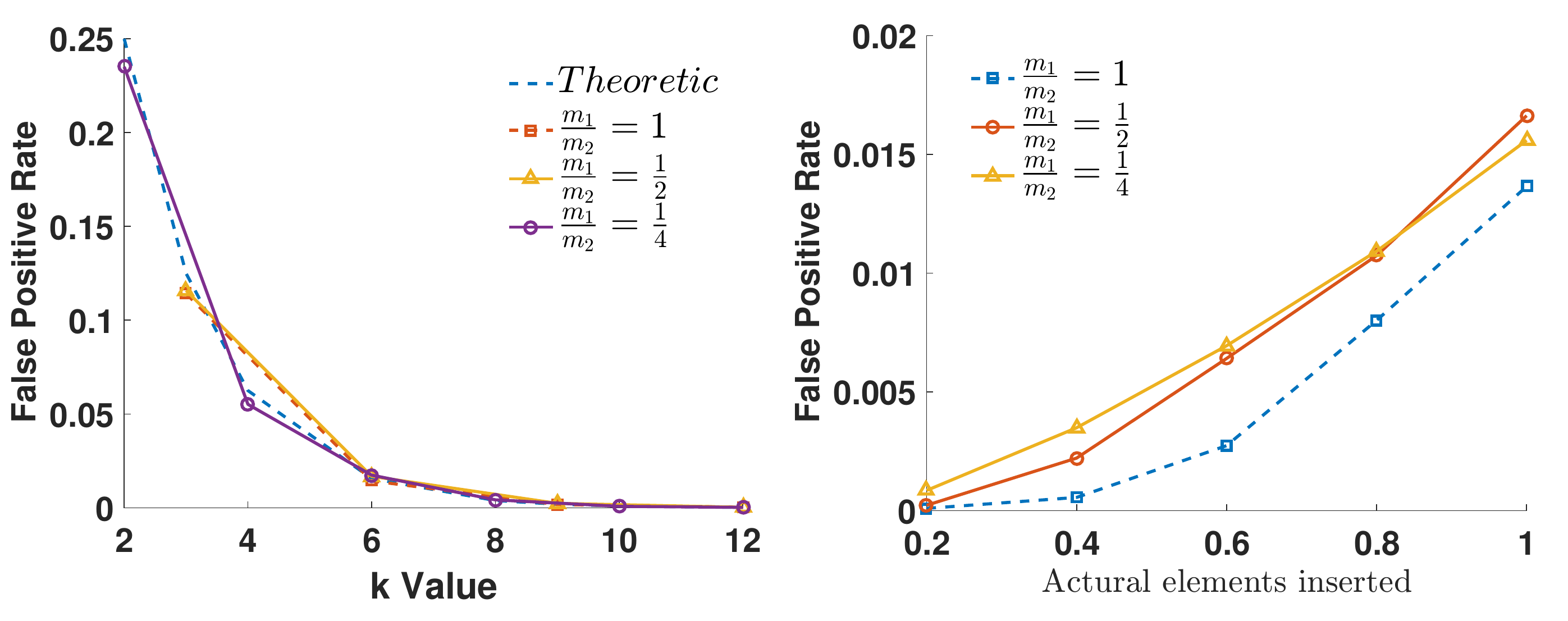}
\caption{False positive rate of maximum adaptive matrix varying with $k$ and actual elements inserted, respective to different $\frac{m_{1}}{m_{2}}$s.}
\label{fig7}
\end{figure}

Finally, let's test the load factor of our matrix with respect to the proportion of inserted elements. The parameters we choose are same to the former experiments. As shown in the table, the load factor of different $\frac{m_{1}}{m_{2}}$s are almost the same under the same proportion of inserted elements. Specially, when the matrix is fully loaded, the load factor is very near to 25\% as we've predicted in the context.

\begin{table}[!h]
\centering
\caption{Load factor \& proportion of elements of different $\frac{m_{1}}{m_{2}}$s}
\begin{tabular}{|c|c|c|c|c|}
\hline
\diagbox{Proportion}{Load factor}{$\frac{m_{1}}{m_{2}}$}&$\frac{1}{4}$&$\frac{1}{2}$&$1$\\
\hline
20\%&0.063681451&0.065117817&0.065625635\\
\hline
40\%&0.121001445&0.120029788&0.122250358\\
\hline
60\%&0.170028163&0.167312289&0.175074846\\
\hline
80\%&0.213055668&0.210467359&0.217504935\\
\hline
100\%&0.249106352&0.24721041&0.256099335\\
\hline
\end{tabular}
\end{table}

\textbf{On minimum storage matrix.} In this part, we evaluate the false positive rates of the minimum storage matrix. Let $k=\frac{k_{1}}{j}=k_{2}$. We fix $n_{1}=jn_{2}$, $n_{2}=144$ and let $k$ be different integers. $m_{2}$s is adjusted to appropriate values with the formula $k^{2}=\frac{m_{2}^{2}}{n_{2}}ln2$.

Let $j=2,10,40,100$ respectively and draw the tendencies of false positive rates varying with respect to $k^{2}$ on the left part of Fig. \ref{fig8}. In this picture, we can see on all the 4 curves our j-matrix work well and behave almost exactly the same whatever the value of $j$ is. Again, as an approximation, we take some points nearby $k^{2}$ where $k$s are actually not integers, however, the number of hash functions are fixed to integers. Similarly, the corresponding amount of bits are calculated from the non-integer values of $k$.

We further fix $k^{2}$ at $4,9,16$ respectively, and test the varying tendencies of false positive rates with respect to $j$. As shown in the picture, it is very steady, thus the result in the right part of Fig. \ref{fig8} is highly robust to the disturbance of $j$.

\begin{figure}[!h]
\centering
\includegraphics[width=1.4in,angle=270]{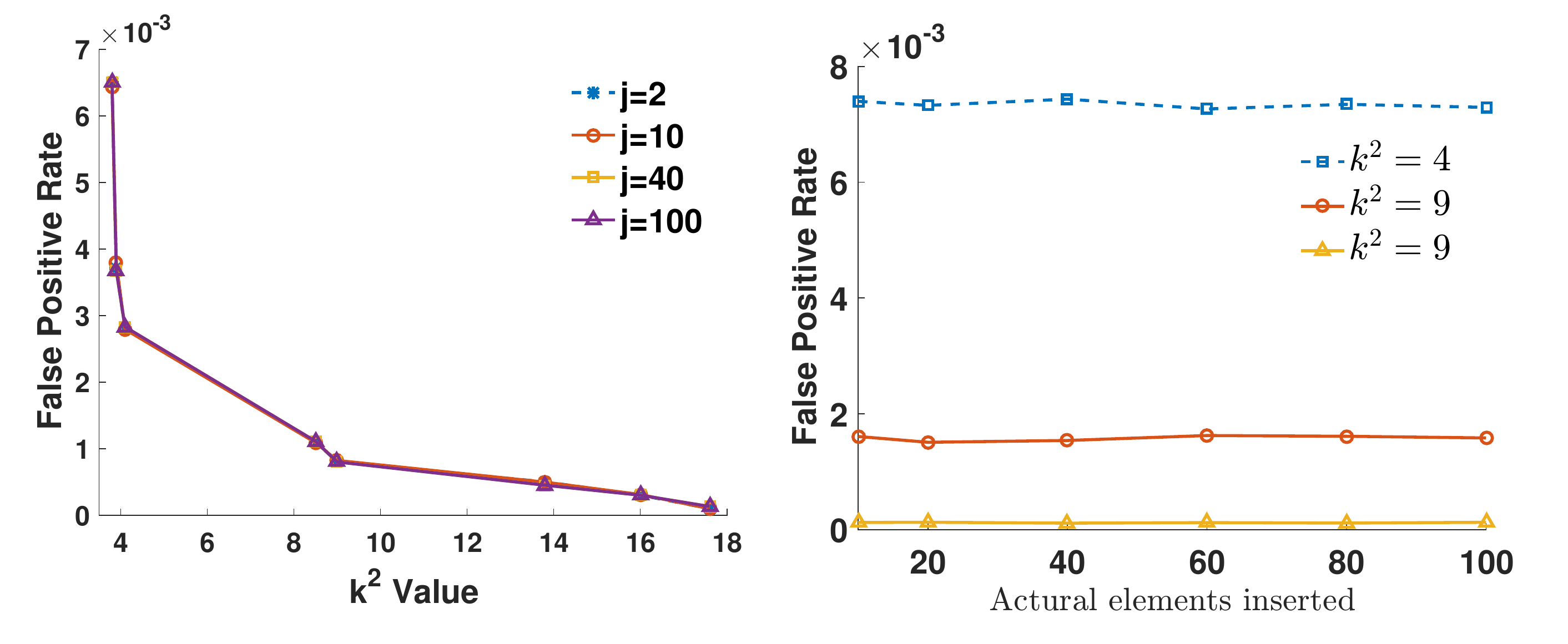}
\caption{False positive rate of $j$-matrix varying with $j$ and $k^{2}$, respective to different $k$s.}
\label{fig8}
\end{figure}

\subsection{2-tuple partial existence tests}In this part, we will evaluate the average comparison times performing batch 2-tuple partial existence tests on different datasets. As a contrast, the performance of both matrix Bloom filter and hashmap will be shown. in the remaining of this section, the number of hash functions in the matrix Bloom filters is fixed to four.

$\bullet$\textbf{On real-world datasets.} Let's perform batch 2-tuple partial existence tests on the mentioned real-world datasets KOS and NIPS, both of which are text collections in the form of bags-of-words. For each text collection, the contents of the first two columns are ``the number of documents" and ``the number of words in the vocabulary", respectively. We treat them as key and value, and perform insertions/queries on the matrix Bloom filter and the hashmap. The length of the hashmap is fixed at a value equal to the number of keys. As described, the lookup dataset involves two sections of data. The $x$ axis represents the proportion the first section occupies. The $y$ axis is the average comparison times. 

Fig. \ref{fig9} shows the tendencies of average comparison times at different length of hashmap varying with respect to the proportion. It involves six different curves, one of them denotes the matrix Bloom filter, while the remaining five denote different amount of data inserted into hashmaps, varying in the range of $1\sim10$ times of the number of keys. From this figure, we can see the matrix Bloom filter holds a constant of average comparison times. The greater proportion of section one data, the worse performance of hashmap is. Besides, more values corresponding to a same key results in more times of comparison, since the chain will be longer and results in a greater expense of traverse.

\begin{figure}[!h]
\centering
\includegraphics[width=1.4in,angle=270]{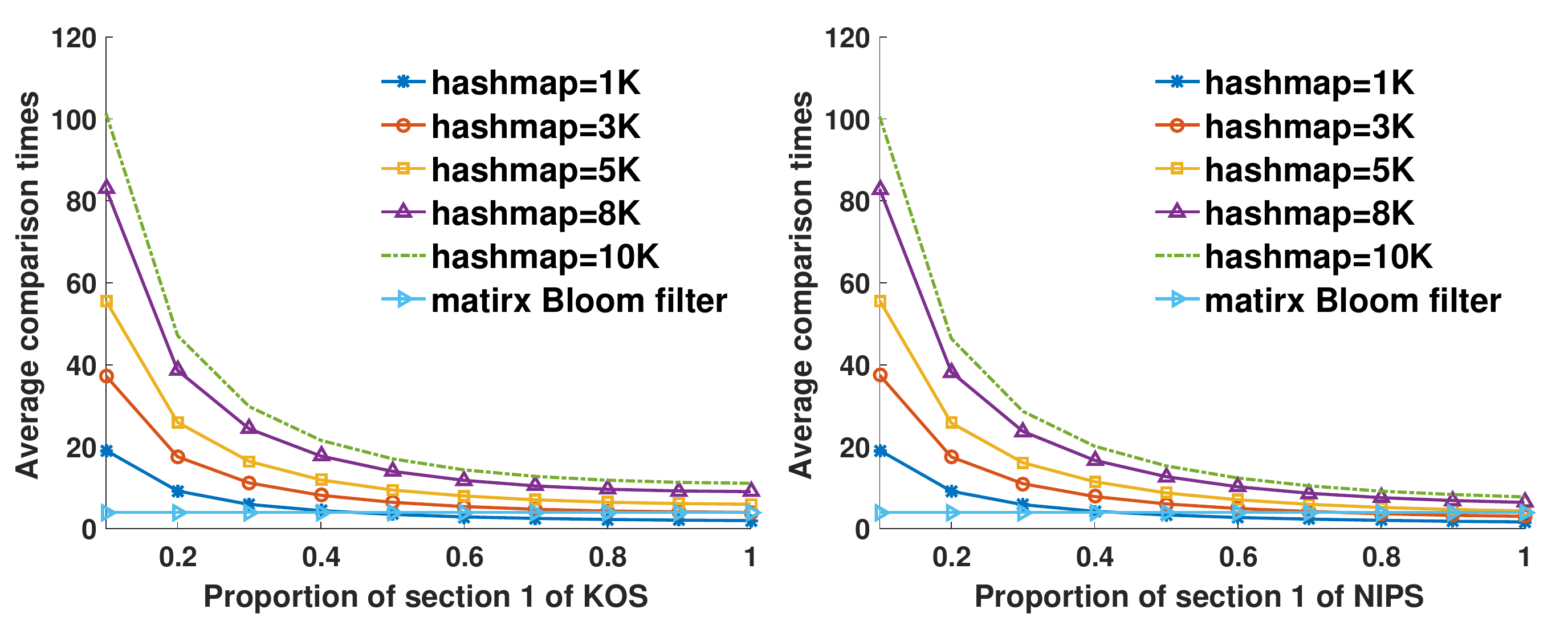}
\caption{2-tuple partial existence tests on hashmaps and matrix Bloom filters (maximum adaptive matrix).}
\label{fig9}
\end{figure}

We also test the runtime of the same experiments, however, in this time, the testing dataset only involves the first section. We generate dataset KOS1, KOS2 and NIPS1, NIP2 as subsets of KOS and NIPS, respectively. The size of the hashmaps are always adjusted to the same number of the keys.

The results are summarised in the following Table \ref{tab3}. Form this table, we can see the runtime of the matrix Bloom filter is almost steady, and is far less than that in hashmaps. As comparison, generally, the longer the chain in hashmap is, the more average time it costs, which is consistent to our theory predicted. Hence, with the scale of data increase, 
it is more economical to employ a matrix Bloom filter to perform batch partial queries on 2-tuples.

\begin{table}[!h]
\centering
\caption{Average runtime of matrix Bloom filter and hashmap, on split dataset KOS and NIPS.}
\label{tab3}
\begin{tabular}{|c|c|c|c|c|c|}
\hline
dataset  & key  & value & \textless{}key, value\textgreater{} & BF ($\mu$s)& Hashmap ($\mu$s)\\
\hline
KOS1  & 50   & 5963  & 25527                              & 1.00          & 14.5        \\
\hline
KOS2  & 100  & 3114  & 9586                               & 1.04          & 2.71         \\
\hline
KOS   & 3430 & 6906  & 353160                             & 1.05         & 2.97       \\
\hline
NIPS1 & 50   & 5851  & 25246                              & 1.02          & 8.2        \\
\hline
NIPS2 & 100  & 7703  & 49732                              & 1.03          & 8.02        \\
\hline
NIPS3 & 1500 & 12375 & 746316                             & 1.01         & 7.97  \\
\hline    
\end{tabular}
\end{table}

$\bullet$\textbf{On full-repeating datasets.} Let's decide the performance of batch 2-tuple partial existence tests on full-repeating datasets and make comparison between the matrix Bloom filter and the hashmap. The results are shown in Fig. \ref{fig10}, with a fixed number of keys (500), and the value varies from 100 to 700. The left part of Fig. \ref{fig10} is the average comparison times under different amount of (key, value) pairs, while the right part is the tendencies of average comparison times varying with respect to the size of hash table\footnote{With dataset (500,100) key-value pairs.}.

\begin{figure}[!h]
\centering
\includegraphics[width=1.4in,angle=270]{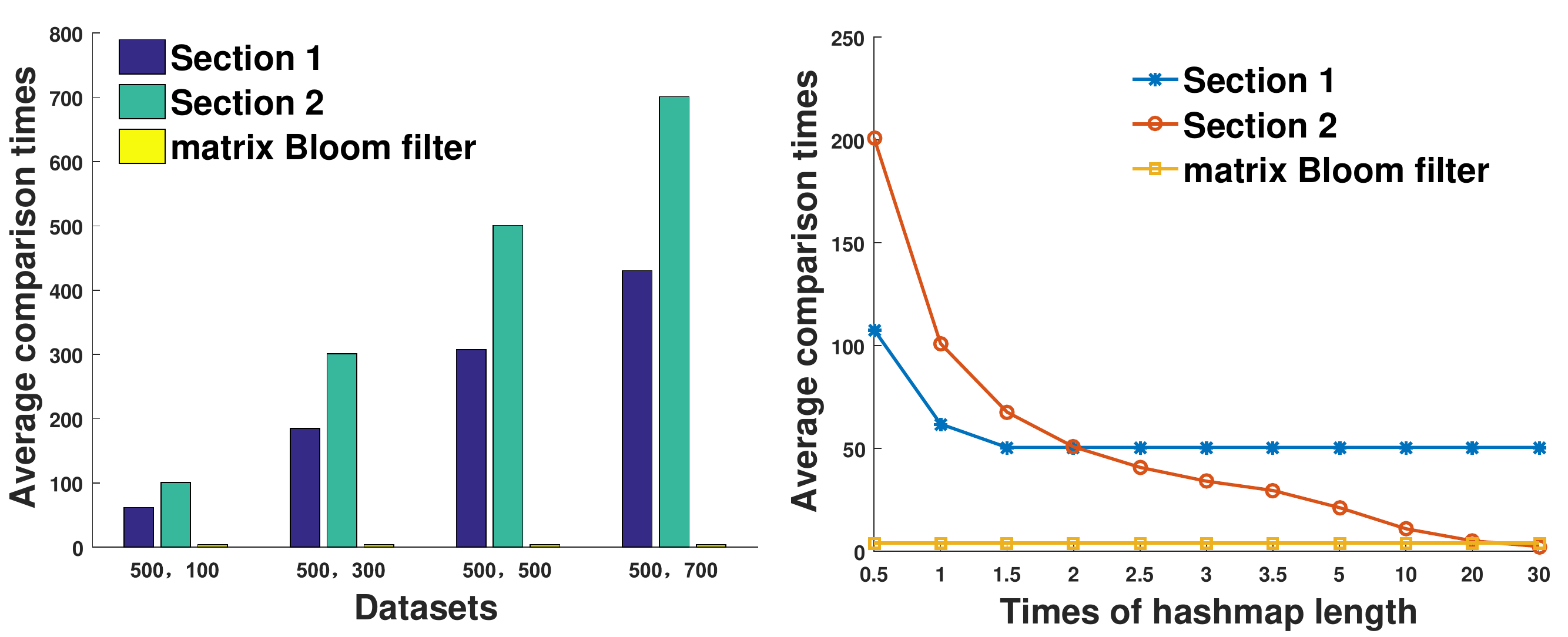}
\caption{The average comparison times on full-repeating datasets, with matrix Bloom filters and hashmaps.}
\label{fig10}
\end{figure}

In the left part of Fig. \ref{fig10}, we can see our matrix Bloom filter also holds a constant comparison time. The right part is the tendency where inserted dataset is fixed, while the length of hashmap varies. The $y$ axis also denotes to the comparison times, and the $x$ axis denotes to the times of number of the key. The two curves represents the mentioned two sections of lookup datasets, respectively. With the increase of hashmap length, dataset section one decreases at first, while keeps steady afterwards, since the collision gradually disappears, and the queries meet the same set of chains. Dataset section two decreases to 1 rapidly, as the queries meet empty buckets (of one time comparison) with higher probability when the hashmap tends to be empty.

\subsection{Double-side partial queries}In this part, we perform double-side partial queries on full-repeating datasets, and test the comparison time and runtime of hashmaps and our matrix Bloom filter. Tuples are inserted in the structures, and the queries are performed in the sequence of (key, value) and (value,key), respectively. The $x$ axis denotes to datasets generated from different numbers of key-value pairs.

As shown in the left part of Fig. \ref{fig11}, the larger the overall dataset is, the more comparison times it holds. Besides, the comparison times of (key, value) sequence is less than that of (value,key) sequence, since the chain of the former is shorter than that of the latter. As comparison, the comparison time of matrix Bloom filter keeps constant. The right part of Fig. \ref{fig11} shows the runtime of the same experiment performed on those datasets, and the same conclusion can be deduced.

\begin{figure}[!h]
\centering
\includegraphics[width=1.4in,angle=270]{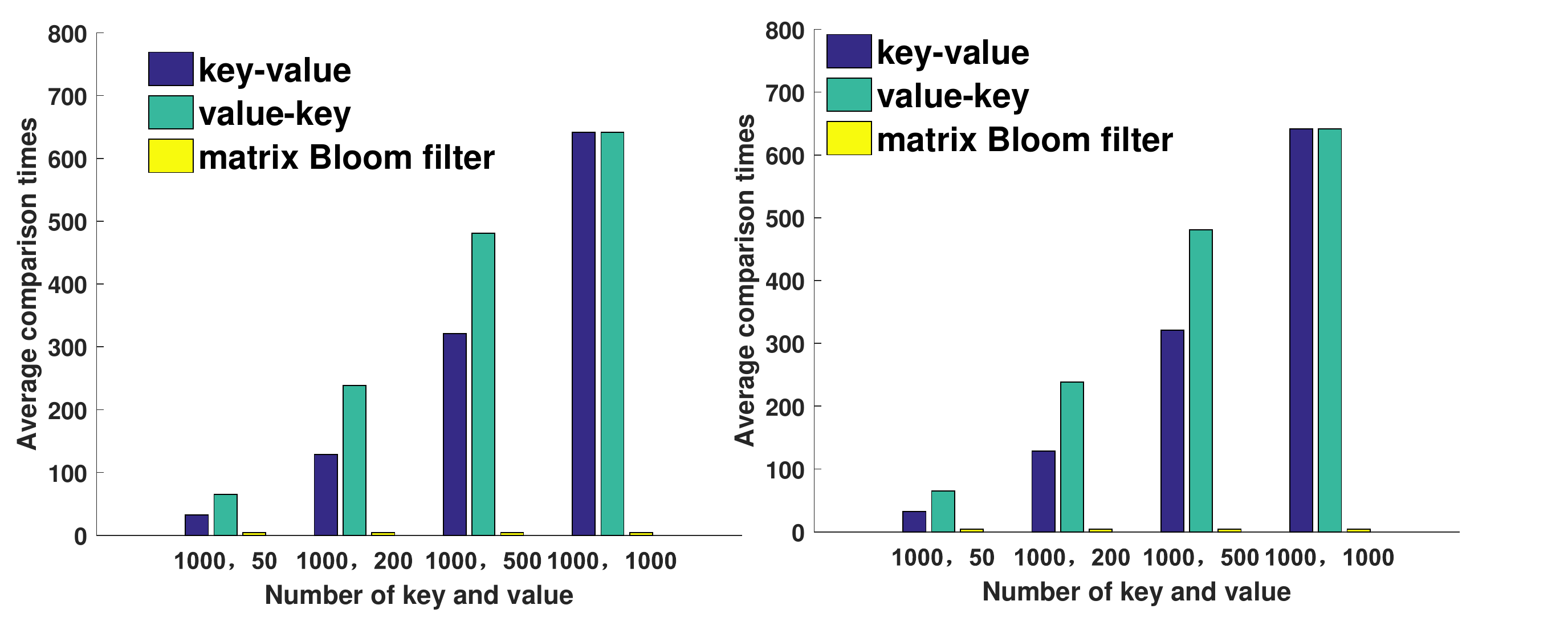}
\caption{Average comparison times and runtime of matrix Bloom filter and hashmap, on dataset KOS and NIPS.}
\label{fig11}
\end{figure}

\section{Related Work}There were some prior works attempted to extend the form of Bloom filter into the matrix form. Geravand et al. proposed a bit matrix to detect copy-paste contents in a literature library\cite{6}. Wang et al. proposed a Bloom filter in a matrix form as a ``real" matrix representation of a new type of Bloom filter\cite{7}. It employs $s$ rows of Bloom filters with $k$ hashes along with an additional special hash to decide which row of Bloom filter should be chosen to insert. In the mentioned works, the proposed data structures are nominally called ``matrix Bloom filter", since the authors formally replace the initial bit-vector in a standard Bloom filter by a bit-matrix. However, they behave like a group of co-working standard Bloom filters rather than a nature extension of Bloom filters to a matrix form. Diversely, the matrix Bloom filter in our paper is truly a higher dimensional version of the standard Bloom filter, where both rows and columns in the matrix work like independent Bloom filters.

Multi-set problem is also close to the 2-tuple batch lookup problem, which can be defined as the following form (See appendix C for more details): Given $v$ sets $U_{1}, U_{2}, ...U_{v}$, where any two of them have no intersection. Given an item $a$, which set does $a$ belong to? Aiming at this problem, Hao et al. proposed a well-known solution, namely, the combinational Bloom filter\cite{11}. It considers the membership query of S which has multiple groups, whose results should tell which group the input element belongs to, without the prior knowledge of the number of elements in each group. The performance is not quite ideal, since a large number of hash functions and memory accesses are required, and the probability of misclassification in combinatorial Bloom filter is high, as false positives on each bit will result in incorrect group code. As another solution, Yu et al. proposed vBF\cite{8} in the framework of Bloom filter. It builds $v$ Bloom filters for the $v$ sets. When any element $a$ is queried, vBF executes query in all the $v$ BFs. If $BF_{i}$ reports true, $a$ is deemed in set $i$, otherwise no. There are also many other Bloom filter-based solutions to the multi-set problem: the coded Bloom filter\cite{9}, the Bloomier filter\cite{10}, kBF\cite{12}.

Partial evaluation is also a paradiam closely related to our matrix Bloom filter, which was first described for a programming language by Yoshihiko Futamura in 1983\cite{13}. It is a technique to produce new programs that run faster than the originals while being guaranteed to behave in the same way. For example, assume that there are two input data $(x_{1},x_{2})$ for a given 2-element function $f$, and $f(x_{1},x_{2})$ turns out to be the computation result of $f$. Suppose the data set to be $\{(1,2),(1,3),(1,4),(2,5),(3,6)\}$. It is easy to see that $x_{1}=1$ appears frequently, hence it is more efficient to pre-generate a new 1-element function $f_{1}(x_{2})$ where $f_{1}(x_{2})=f(1,x_{2})$, and calculate the remaining result by substituting all $x_{2}s$ into $f_{1}$. Our matrix Bloom filter is a specific partial evaluation algorithm in the context of 2-tuple query, whose meaning and functionality is more than ``efficiency".

As a space-efficient probabilistic data structure to support membership queries on a set of elements, Bloom filters are generally accepted to be useful for cases where space is limited and membership query is required with allowable false positives. There are extensive applications of Bloom filters in the area of databases, e.g. query and search\cite{14,15}, privacy preservation\cite{16,17}, key-value store\cite{17,18}, content synchronization\cite{19,20,21,22,23}. Bloom filters can also find their applications in a wide range of networking\cite{24,25,26,27,28,29,30,31,32,33,34,35}.

\section{Conclusion}In this paper, we propose a novel framework of matrix Bloom filter as well as its two variants, i.e. the maximum adaptive matrix and the minimum storage matrix, as a high-dimensional extension of the standard Bloom filter. Through theoretical and empirical analysis, we prove and show it can efficiently process 2-tuple batch insertions and queries. 

\section*{Acknowledgement}Corresponding author: Dagang Li. Yue Fu and Rong Du contributed equally to this paper.

\section*{Appendix}
\subsection{Proof of theorem \ref{theorem1}}
\begin{proof}
Suppose the hashes choose the positions equiprobably. In a single hash insertion, for any bit the probability that it is not set to 1 turns out to be:

\begin{equation}\label{1}
p_{1}=1-\frac{1}{m_{1}m_{2}}
\end{equation}

After all the $k_{1}k_{2}$ executions the probability that the bit is not set to 1 is:

\begin{equation*}
p_{2}=(1-\frac{1}{m_{1}m_{2}})^{k_{1}k_{2}}
\end{equation*}

When $n$ elements are inserted, the probability that the bit is still 0 is:

\begin{equation*}
p_{3}=(1-\frac{1}{m_{1}m_{2}})^{nk_{1}k_{2}}
\end{equation*}

Hence, the the probability that the bit is set to 1 is:

\begin{equation*}
p_{4}=1-(1-\frac{1}{m_{1}m_{2}})^{nk_{1}k_{2}}
\end{equation*}

The false positive result occurs when all the queried $k_{1}k_{2}$ bits are set to 1:

\begin{equation*}
p_{5}=(1-(1-\frac{1}{m_{1}m_{2}})^{nk_{1}k_{2}})^{k_{1}k_{2}} \approx (1-e^{\frac{-nk_{1}k_{2}}{m_{1}m_{2}}})^{k_{1}k_{2}}
\end{equation*}
\end{proof}

\subsection{Proof of theorem 2}
\begin{proof}
Notice  

\begin{equation}
p_{5}=e^{k_{1}k_{2}ln(1-e^{\frac{-nk_{1}k_{2}}{m_{1}m_{2}}})}
\end{equation}

Let 

\begin{equation*}
p=e^{\frac{-nk_{1}k_{2}}{m_{1}m_{2}}} 
\end{equation*}

\begin{equation*}
g=k_{1}k_{2}ln(1-e^{\frac{-nk_{1}k_{2}}{m_{1}m_{2}}})=\frac{-m_{1}m_{2}}{n} lnpln(1-p)
\end{equation*}

 Clearly, $p_{5}$ arrives at its minimum when $g$ reaches its minimum. Due to symmetry, the following restriction holds: 
\begin{equation}\label{2}
p=\frac{1}{2}
\end{equation}

It means when 50\% bits are occupied with 1, the minimum is reached. Hence, we have the boundary condition:
\begin{equation*}
k_{1}k_{2}=\frac{m_{1}m_{2}}{n}ln2
\end{equation*}

As well as the minimum value of $p_{5}$:

\begin{equation*}
p_{min}=(\frac{1}{2})^{k_{1}k_{2}}
\end{equation*}

Compare with the corresponding results of the standard Bloom filter:

\begin{equation}\label{3}
p=\frac{1}{2}, k=\frac{m}{n}ln2, p_{min}=(\frac{1}{2})^{k}
\end{equation}

Replace $k$ and $m$ in eq.\ref{3} by $k_{1}k_{2}$ and $m_{1}m_{2}$, then complete the proof.
\end{proof}

\subsection{On multi-set query problems}In the related work, we have discussed the notion of multi-set query. We further illustrate it is naturally a 2-tuple partial existence test problem that could be solved by our matrix Bloom filter. Let's start with vBF as a baseline solution, which is in fact a special case of our matrix Bloom filter. Furthermore, we will define the multi-set query problems in a more general form, and conclude it into the framework of our matrix Bloom filter.

$\bullet$\textbf{vBF: special multi-set query.} Notice that multi-set problem can be induced into 2-tuple partial existence test problem, where 2-tuples are denoted by (element name, set number). We say vBF is in fact a special case of our matrix Bloom filter: vBF supports 2-tuple partial existence tests, while does not support double-side partial queries. During insertions/queries, we at first find out the corresponding set number that the element name belongs to, and then inserts/queries that element name into a standard Bloom filter. In our matrix Bloom filter, this is realised by a series of row hash functions that map the set number to the corresponding Bloom filters yet to insert the element name. 

However, in the scene of vBF, it has been priorly aware that the involved set numbers are determined and pre-given. Hence, a single special row hash function can be used to map set number $i$ to the $i$th BF, thus it is convenient to answer such question ``given elements $a,b,c,d$, do they belong to set 2?" We can simply find out the set number we need, and check the element name into the corresponding Bloom filters. However, the opposite queries on a fixed element name is not allowed in the same pattern: Relatively, a vBF can only query an 2-dimensional element of sequence (set number, element name). When question comes in the form of ``given an element $a$, which set does it belong to?", it can only traverse all the BFs to hunt for that answer.

All in all, we can define the special multi-set problem in a more precise form: given $v$ \textbf{pre-determined} sets and an element $e$, which set does $e$ belong to? No doubt that vBF is designed for this scene. Of course matrix Bloom filter can work in a special case (as a vBF) to deal with such special multi-set problems.

$\bullet$\textbf{Matrix Bloom filter: General multi-set query.} It can be seen from the previous discussions, our matrix Bloom filter is adaptive to a more general case where the involved set numbers remain undetermined. In other words, the set number itself becomes an object yet to be searched in a range, just like the element name. In this scenario, it is required to support double-side queries. To make it clear, we propose the following definition of general multi-set membership test problem:

\begin{definition}
General multi-set membership test problem. Given a series of sets $\mathbb{A}_{1},\mathbb{A}_{2},...\mathbb{A}_{x} \in \mathbb{U}$, where $1,2,...x\in\mathbb{X}$. Given a series of  element names $e_{i}\in \mathbb{A}_{\gamma_{i}}$, where $\gamma_{i}\in\mathbb{X}$. Then, $\mathbb{B}=\cup\mathbb{A}_{\gamma_{i}}$ is a proper subset of $\mathbb{U}$ that \textbf{remains unknown} in the process.\\
1. Which set does $e_{i}$ belong to?\\
2. For set $\mathbb{A}_{\gamma_{i}}$ what's the name of all element inserted in it?
\end{definition}

Clearly, a vBF works reluctantly for this problem if we let $v=x$, which means $x$ standard Bloom filters of the same length shall be pre-determined. In other words, the overall size is relevant to the scale of $\mathbb{X}$. Suppose $n_{1}$ elements and the corresponding set numbers are to inserted into a matrix Bloom filter. Let $|\mathbb{B}|$ be $n_{2}$. Then, $max(n_{1},n_{2})$ will be the number of the 2-dimensional elements (set number, element name). Hence, the overall size of matrix Bloom filter will be relevant to the scale of max($|\mathbb{B}|$, $n_{1}$), which may be more space efficient when the involved set numbers are rather small in the range of all possible values.

\end{document}